\documentclass[11pt]{jmlr}

\usepackage{amsmath, amssymb, algorithmic,algorithm,graphicx, mathtools}
\usepackage{tikz, pgfplots}

\begin{document}
\author{\Name{Nicholas Ruozzi} \Email{nicholas.ruozzi@epfl.ch}\\
       \addr Communication Theory Laboratory\\
       \'{E}cole Polytechnique F\'{e}d\'{e}rale de Lausanne\\
       Lausanne 1015, Switzerland
       \AND
       \Name{Sekhar Tatikonda} \Email{sekhar.tatikonda@yale.edu}\\
       \addr Department of Electrical Engineering\\
       Yale University\\
       New Haven, CT 06520, USA}

\title[Message-Passing Algorithms for Quad. Minimization]{Message-Passing Algorithms for Quadratic Minimization}
\maketitle

\begin{abstract}
Gaussian belief propagation (GaBP) is an iterative algorithm for computing the mean of a multivariate Gaussian distribution, or equivalently, the minimum of a multivariate positive definite quadratic function.  Sufficient conditions, such as walk-summability, that guarantee the convergence and correctness of GaBP are known, but GaBP may fail to converge to the correct solution given an arbitrary positive definite quadratic function.  As was observed by \citet{malioutov}, the GaBP algorithm fails to converge if the computation trees produced by the algorithm are not positive definite.  In this work, we will show that the failure modes of the GaBP algorithm can be understood via graph covers, and we prove that a parameterized generalization of the min-sum algorithm can be used to ensure that the computation trees remain positive definite whenever the input matrix is positive definite.  We demonstrate that the resulting algorithm is closely related to other iterative schemes for quadratic minimization such as the Gauss-Seidel and Jacobi algorithms.  Finally, we observe, empirically, that there always exists a choice of parameters such that the above generalization of the GaBP algorithm converges.
\end{abstract}


\section{Introduction}
In this work, we study the properties of reweighted message-passing algorithms with respect to the quadratic minimization problem.  Let $\Gamma\in \mathbb{R}^{n\times n}$ be a symmetric positive definite matrix and $h\in\mathbb{R}^{n}$.  The quadratic minimization problem is to find the $x\in\mathbb{R}^n$ that minimizes $f(x) = \frac{1}{2}x^T\Gamma x - h^Tx$.  Minimizing a positive definite quadratic function is equivalent to computing the mean of a multivariate Gaussian distribution with a positive definite covariance matrix or equivalently, solving the positive definite linear system $\Gamma x = h$ for the vector $x$.  

Gaussian belief propagation (GaBP), is an iterative message-passing scheme that can be used to estimate the mean of a Gaussian distribution as well as individual variances.  Because of their distributed nature and their ability to provide estimates of the individual means and variances for each variable, using the GaBP and min-sum algorithms to solve linear systems has been an active area of research.  

In previous work, several authors have provided sufficient conditions for the convergence of GaBP.  \citet{weissgauss} demonstrated that GaBP converges in the case that the covariance matrix is diagonally dominant. \citet{malioutov} proved that the GaBP algorithm converges when the covariance matrix is walk-summable. \citet{quadciamac,convexciamac} showed that scaled diagonal dominance was a sufficient condition for convergence and also characterized the rate of convergence via a computation tree analysis.  The later two sufficient conditions, walk-summability and scaled diagonal dominance, are known to be equivalent \citep{malthesis, allerton09}.  

While the above conditions are sufficient for the convergence of the GaBP algorithm they are not necessary: there are examples of positive definite matrices that are not walk-summable for which the min-sum algorithm still converges to the correct solution \citep{malioutov}.  A critical component of these examples is that the computation trees remain positive definite throughout the algorithm.  Such behavior is guaranteed if the original matrix is scaled diagonally dominant, but arbitrary positive definite matrices can produce computation trees that are not positive definite \citep{malioutov}.  If this occurs, the standard GaBP algorithm fails to produce the correct solution.

One proposed solution to the above difficulties is to precondition the covariance matrix in order to force it to be scaled diagonally dominant.  Diagonal loading was proposed as one such useful preconditioner in \citet{johnsonfix}.  The key insight of diagonal loading is that scaled diagonal dominance can be achieved by sufficiently weighting the diagonal elements of the covariance matrix.  The diagonally loaded matrix can then be used as an input to a GaBP subroutine.  The solution produced by GaBP is then used in a feedback loop to produce a new matrix, and the process is repeated until a desired level of accuracy is achieved.  Unfortunately, this approach results in an algorithm that, unlike GaBP, is not decentralized and distributed choosing the appropriate amount of diagonal loading and feedback to achieve quick convergence remains an open question. 

Recent work has studied provably convergent variants of the min-sum algorithm.  The result has been the development of many different ``convergent and correct'' message-passing algorithms:  MPLP \citep{MPLP}, max-sum diffusion \citep{MSD}, norm-product belief propagation \citep{hazan}, and tree-reweighted belief propagation \citep{waintree}.  Each of these algorithms can be viewed as a coordinate ascent/descent scheme for an appropriate lower/upper bound.  A general overview of these techniques and their relationship to bound maximization can be found in \citet{sontag} and \citet{weissconv}.  These algorithms guarantee convergence under an appropriate message-passing schedule, and they also guarantee correctness if a unique assignment can be extracted upon convergence.  Such algorithms are a plausible candidates in the search for convergent GaBP style message-passing algorithms.

The primary contributions of this work are twofold.  First, we demonstrate that graph covers can be used to provide a new, combinatorial characterization of walk-summability.  This characterization allows us to conclude that ``convergent and correct message-passing'' schemes based on dual optimization techniques that guarantee the correctness of locally decodable beliefs cannot converge to the correct minimizing assignment outside of walk-summable models.

Second, we investigate the behavior of reweighted message-passing algorithms for the quadratic minimization problem.  The motivation for this study comes from the observation that belief propagation style algorithms typically do not explore all nodes in the factor graph with the same frequency \citep{frey}.  In many application areas such uneven counting is undesirable and typically results in incorrect answers, but if we can use reweighting to overestimate the diagonal entries of the computation tree relative to the off diagonal entries, then we may be able to force the computation trees to be positive definite at each iteration of the algorithm.  Although similar in spirit to diagonal loading, our approach preserves the distributed message-passing structure of the algorithm.  We will show that there exists a choice of parameters for the reweighted algorithms that guarantees monotone convergence of the variance estimates on all positive definite models, even those for which the GaBP algorithm fails to converge.  We empirically observe that there exists a choice of parameters that also guarantees the convergence of the mean estimates.

In addition, we show that our graph cover analysis extends to other iterative algorithms for the quadratic minimization problem, and that similar ideas can be used to reason about the min-sum algorithm for general convex minimization.

The outline of this paper is as follows: In Section \ref{sec:prelim} we review the min-sum algorithm, its reweighted generalizations, and the quadratic minimization problem.  In Section \ref{sec:gc} we discuss the relationship between pairwise message-passing algorithms and graph covers, and we show how to use graphs covers to characterize walk-summability.  In Section \ref{sec:convmsg}, we examine the convergence of the means and the variances under the reweighted algorithm for the quadratic minimization problem, we examine the relationship between the reweighted algorithm and the Gauss-Seidel and Jacobi methods, and we compare the performance of the reweighted algorithm to the standard min-sum algorithm.  Finally, in Section \ref{sec:conc}, we summarize the results and discuss extensions of this work to general convex functions  as well as open problems.  Detailed proofs of the two main theorems can be found in the Appendices \ref{ap:2cover} and \ref{ap:posdefcomp}.

\section{Preliminaries}
\label{sec:prelim}
In this section, we review the min-sum algorithm and a reweighted variant over pairwise factor graphs.  Of particular importance for later proofs will be the computation trees generated by each of these algorithms.  We also review the quadratic minimization problem, and derive the closed form message updates for this problem.

\subsection{The Min-Sum Algorithm}
\label{sec:minsum}
The min-sum algorithm attempts to compute the
minimizing assignment of an objective function $f:
\prod_i\mathcal{X}_i\rightarrow \mathbb{R}$ that, given a graph $G=(V,E)$, can be factorized as a sum of
self-potentials and edge potentials as follows: $f(x_1,...,x_n) = \sum_{i\in V}
\phi_i(x_i) + \sum_{(i,j)\in E} \psi_{ij}(x_i, x_j)$.  We assume that this minimization problem is well-defined:  $f$ is bounded from below and there exists an $x\in \prod_i\mathcal{X}_i$ that minimizes $f$.

To each factorization, we associate a bipartite graph known as the factor graph.  In general, the factor graph consists of a node for each of the variables  $x_1,..., x_n$, a node for each of the $\psi_{ij}$, and for all $(i,j)\in E$, an edge joining the  node corresponding to $x_i$ to the node corresponding to $\psi_{ij}$.  Because the $\psi_{ij}$ each depend on exactly two factors, we often omit the factor nodes from the factor graph construction and replace them with a single edge.  This reduces the factor graph to the graph $G$.  See Figure \ref{factfig} for an example of this construction.

\begin{figure}
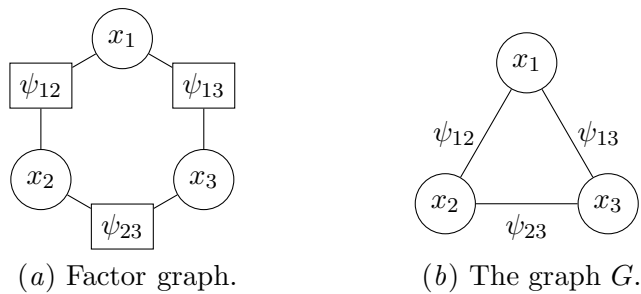

\centering
\subfigure[{Factor graph.}]{\includeteximage[]{fig_G.txt}}
\hspace{2cm}
\subfigure[The graph $G$.]{\includeteximage[]{fig_fact.txt}}
\caption[Example factor graph.]{The factor graph corresponding to $f(x_1,x_2,x_3) = \phi_1 + \phi_2 +
\phi_3 + \psi_{12} + \psi_{23} + \psi_{13}$ and the graph $G$. The functions $\phi_1$,$\phi_2$, and $\phi_3$ each depend only on variable, and are typically omitted from the factor graph representation for clarity.}
\label{factfig}
\end{figure}

We can write the min-sum algorithm as a local message-passing algorithm over the graph $G$.  During the execution of the min-sum algorithm, messages are passed back and forth between adjacent nodes of the graph.  On the $t^{th}$ iteration of the algorithm, messages are passed along each edge of the factor graph as
\begin{eqnarray*}
m^t_{i\rightarrow j}(x_j) & = & \kappa + \min_{x_i} \psi_{ij}(x_i,
x_j) + \phi_i(x_i)+\sum_{k \in \partial i \setminus j}m^{t-1}_{k\rightarrow i}(x_i),
\end{eqnarray*}
where $\partial i$ denotes the set of neighbors of variable node $x_i$ in the factor graph and $\partial j \setminus i$ is abusive notation for the set-theoretic difference $\partial j \setminus \{i\}$.  When the factor graph is a tree, these updates are guaranteed to converge, but understanding when these updates converge to the correct solution for an arbitrary graph is a central question underlying the study of the min-sum algorithm.

Each message update has an arbitrary normalization factor $\kappa$.  Because $\kappa$ is not a function of any of the variables, it only affects the value of the minimum and not where the minimum is located.  As such, we are free to choose it however we like for each message and each time step.  In practice, these constants are used to avoid numerical issues that may arise during the execution of the algorithm.  

We will think of the messages as a vector of functions indexed by the edge over which the message is passed.  Any vector of real-valued messages is a valid choice for the vector of initial messages $m^0$, and the choice of initial messages can greatly affect the behavior of the algorithm.  A typical assumption, that we will use in this work, is that the initial messages are chosen such that $m^0_{i\rightarrow j} \equiv 0$ for all $i$ and $j$.

We can use the messages in order to construct an estimate of the min-marginals of $f$.  Given any vector of messages, $m^t$, we can construct a set of beliefs that are intended to approximate the min-marginals of $f$ as
\begin{eqnarray*}
\tau^t_i(x_i) & = & \kappa + \phi_i(x_i) + \sum_{j \in \partial i} m^t_{j \rightarrow i}(x_i)\label{varb}\\
\tau^t_{ij}(x_i, x_j) & = & \kappa + \psi_{ij}(x_i, x_j) + \tau^t_j(x_j)  - m^t_{i \rightarrow j}(x_j) + \tau^t_i(x_i)  - m^t_{j \rightarrow i}(x_i).
\end{eqnarray*}

Additionaly, we can approximate the optimal assignment by computing an estimate of the argmin,
\begin{eqnarray*}
x_i^t & \in & \arg \min_{x_i} \tau^t_i(x_i).
\end{eqnarray*}

If the beliefs corresponded to the true min-marginals of $f$ (i.e., $\tau^t_i(x_i) = \min_{x' : x'_i = x_i} f(x')$), then for any $y_i\in\arg\min_{x_i} \tau^t_i(x_i)$ there exists a vector $x^*$ such that $x^*_i = y_i$ and $x^*$ minimizes the function $f$.  If $|\arg\min_{x_i} \tau^t_i(x_i)| = 1$ for all $i$, then we can take $x^* = y$, but, if the objective function has more than one optimal solution, then we may not be able to construct such an $x^*$ so easily.  

\begin{definition}
A vector, $\tau = (\{\tau_i\}, \{\tau_{ij}\})$, of beliefs is \textbf{locally decodable} to $x^*$ if $\tau_i(x_i^*) < \tau_i(x_i)$ for all $i$, $x_i\neq x_i^*$.  Equivalently, for each $i\in V$, $\tau_i$ is uniquely minimized at $x_i^*$.
\end{definition}

If the algorithm converges to a vector of beliefs that are locally decodable to $x^*$, then we hope that the vector $x^*$ is a global minimum of the objective function.  This is indeed the case when the factor graph contains no cycles \citep{wainwright} but need not be the case for arbitrary graphical models.

\subsubsection{Computation Trees}
\label{comptree}
An important tool in the analysis of the min-sum algorithm is the notion of a computation tree.  Intuitively, the computation tree is an unrolled version of the original graph that captures the evolution of the messages passed by the min-sum algorithm needed to compute the belief at time $t$ at a particular node of the factor graph.  Computation trees describe the evolution of the beliefs over time, which, in some cases, can help us prove correctness and/or convergence of the message-passing updates.  For example, the convergence of the min-sum algorithm on graphs containing a single cycle can be demonstrated by analyzing the computation trees produced by the min-sum algorithm at each time step \citep{weisslocal}.

The depth $t$ computation tree rooted at node $i$ contains all of the length $t$ non-backtracking walks in the factor graph starting at node $i$.  A walk is non-backtracking if it does not go back and forth successively between two vertices. For any node $v$ in the factor graph, the computation tree at time $t$ rooted at $i$, denoted by $T_i(t)$, is defined recursively as follows: $T_i(0)$ is just the node $i$, the root of the tree. The tree $T_i(t)$ at time $t>0$ is generated from $T_i(t-1)$ by adding to each leaf of $T_i(t -1)$ a copy of each of its neighbors in $G$ (and the corresponding edge), except for the neighbor that is already present in $T_i(t-1)$.  Each node of $T_i(t)$ is a copy of a node in $G$, and the potentials on the nodes in $T_i(t)$, which operate on a subset of the variables in $T_i(t)$, are copies of the potentials of the corresponding nodes in $G$.  The construction of a computation tree for the graph in Figure \ref{factfig} is pictured in Figure \ref{comptreefig}.  Note that each variable node in $T_i(t)$ represents a distinct copy of some variable $x_j$ in the original graph.

\begin{figure}
\centering
  \begin{tikzpicture}[scale=.8]
	\tikzstyle{every node}=[draw,shape=circle, scale=.9, minimum size = 2.5em];	
	\path (0,0) node (X0) {$x_1$};
	\path (-1,-2) node (X2) {$x_2$};	
	\path (-1,-4) node (X4) {$x'_3$};	
	\path (1,-2) node (X6) {$x_3$};	
	\path (1,-4) node (X8) {$x'_2$};			
	\tikzstyle{every node}=[draw=none];	
			
	\path (X0) edge node[label = left:$\psi_{12}$] {} (X2);	
	\path (X0) edge node[label = right:$\psi_{13}$] {} (X6);
	\path (X6) edge node[label = right:$\psi_{23}$] {} (X8);
	\path (X2) edge node[label = left:$\psi_{23}$] {} (X4);	
	\end{tikzpicture}

\caption[Example of a computation tree.]{The computation tree at time $t = 2$ rooted at the variable node $x_1$ of the graph in Figure \ref{factfig}.   The self-potentials corresponding to each variable node are given by the subscript of the variable.}
\label{comptreefig}
\end{figure}
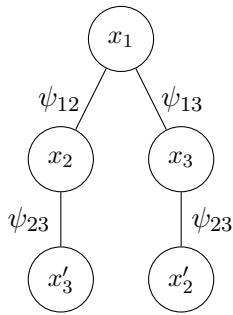

Given any initialization of the messages, $T_i(t)$ captures the information available to node $i$ at time $t$. At time $t=0$, node $i$ has received only the initial messages from its neighbors, so $T_i(0)$ consists only of $i$. At time $t=1$, $i$ receives the round one messages from all of its neighbors, so $i$'s neighbors are added to the tree. These round one messages depend only on the initial messages, so the tree terminates at this point.  By construction, we have the following lemma.

\begin{lemma} 
The belief at node $i$ produced by the min-sum algorithm at time $t$ corresponds to the exact min-marginal at the root of $T_i(t)$ whose boundary messages are given by the initial messages.
\end{lemma}
\begin{proof}
See, for example, \citet{tatjor02, weisscomp}.
\end{proof}

Computation trees provide us with a dynamic view of the min-sum algorithm. After a finite number of time steps, we hope that the beliefs on the computation
trees stop changing and that the message vector converges to a fixed point of the message update equations (in practice, when the beliefs change by less than some
small amount, we say that the algorithm has converged). For any real-valued objective function $f$ (i.e., $|f(x)| < \infty$ for all $x$), there always exists a fixed point of the message update equations (see Theorem 2 of \citet{wainwright}).

\subsection{Reweighted Message-Passing Algorithms}
\label{sec:split}
Because the min-sum algorithm is not guaranteed to converge and, even it does, is not guaranteed to compute the correct minimizing assignment, research has focused on the design of alternative message-passing schemes that do not suffer from these drawbacks.  Efforts to produce provably convergent message-passing schemes have resulted in the rewighted message-passing algorithm described in Algorithm \ref{alg:pairs}.  Notice that if we set $c_{ij} = 1$ for all $i$ and $j$, then we obtain the standard min-sum algorithm.  In \citet{waintree}, the $c_{ij}$ are chosen in a specific way in order to guarantee correctness of the algorithm (which they call TRMP in this special case).  In this work, we will focus on choices of these weights that will guarantee convergence of the algorithm for the quadratic minimization problem.  These choices will, surprisingly, not coincide with those of the TRMP algorithm.  In fact, the choice of weights that guarantees correctness of the TRMP algorithm must necessarily cause the algorithm to either not converge or converge to the incorrect solution whenever the given matrix is not walk-summable.  
\begin{algorithm}[t]
\caption{Synchronous Reweighted Message-Passing Algorithm\label{alg:pairs}}
\begin{algorithmic}[1]
\STATE Initialize the messages to some finite vector.

\STATE For iteration $t = 1,2,...$ update the the messages as follows:
\begin{align*}
m^t_{i\rightarrow j}(x_j) \coloneqq & \kappa + \min_{x_i} \frac{\psi_{ij}(x_i,
x_j)}{c_{ij}} + \phi_i(x_i) + (c_{ij}-1)m^{t-1}_{j\rightarrow i}(x_i) + \sum_{k \in \partial i \setminus j}c_{ki} m^{t-1}_{k\rightarrow i}(x_i).
\end{align*}
\end{algorithmic}
\end{algorithm}

The beliefs for this algorithm are defined analogously to those for the standard min-sum algorithm.
\begin{eqnarray*}
\tau^t_i(x_i) & = & \kappa + \phi_i(x_i) + \sum_{j \in \partial i} c_{ji}m^t_{j \rightarrow i}(x_i)\\
\tau^t_{ij}(x_i, x_j) & = & \kappa + \frac{\psi_{ij}(x_i, x_j)}{c_{ij}} + \tau^t_j(x_j)  - m^t_{i \rightarrow j}(x_j) + \tau^t_i(x_i)  - m^t_{j \rightarrow i}(x_i)
\end{eqnarray*}

The vector of messages at any fixed point of the message update equations has two important properties.  First, the beliefs corresponding to these messages provide an alternative factorization of the objective function $f$.  Second, the beliefs correspond to approximate marginals.  

\begin{lemma}
For any $m^t\in\mathbb{R}^{2|E|}$, 
\[f(x_1,\ldots,x_{|V|}) = \kappa + \sum_{i\in V} \tau_i(x_i) + \Big[\sum_{(i,j)\in E} \tau_{ij}(x_i,x_j) - \tau_i(x_i) - \tau_j(x_j)\Big].\]
\label{lem:admis}
\end{lemma}

\begin{lemma}
If $\tau$ is a set of beliefs corresponding to a fixed point of the message updates in Algorithm \ref{alg:pairs}, then
\[\min_{x_j} \tau_{ij}(x_i,x_j) = \kappa + \tau_i(x_i)\]
for all $(i,j)\in G$ and all $x_i$.
\label{lem:consis}
\end{lemma}

The proof of these two lemmas is a straightforward exercise in applying the definitions.  Similar results for the special case of the max-product algorithm can be found in \citet{wainwright}.

\subsubsection{Computation Trees}
\begin{figure}
\centering
  \begin{tikzpicture}[xscale=1.5, yscale = .9]
	\tikzstyle{every node}=[draw,shape=circle, scale=1, minimum size = 2.5em];	
	\path (0,0) node[label=above:$\phi_1$] (X0) {$x_1$};
	\path (-2,-2) node[label=left:$c_{12}\phi_2$] (X2) {$x_2$};	
	\path (-1,-4) node[label=left:$c_{12}c_{23}\phi_3$] (X4) {$x'_3$};	
	\path (2,-2) node[label=right:$c_{13}\phi_3$] (X6) {$x_3$};	
	\path (1,-4) node[label=right:$c_{13}c_{23}\phi_2$] (X8) {$x'_2$};			

	\path (-3,-4) node[label=left:$(c_{12}^2-c_{12})\phi_1$] (X10) {$x'_1$};	
	\path (3,-4) node[label=right:$(c_{13}^2-c_{13})\phi_1$] (X12) {$x''_1$};	

	\tikzstyle{every node}=[draw=none];	
			
	\path (X0) edge node[label = left:$\psi_{12}$] {} (X2);	
	\path (X0) edge node[label = right:$\psi_{13}$] {} (X6);
	\path (X2)  edge node[label = right:$\frac{\psi_{23}c_{12}}{c_{13}}$] {} (X4);
	\path (X2)  edge node[label = left:$\frac{\psi_{12}(c_{12}-1)}{c_{12}}$] {} (X10);
	\path (X6)  edge node[label = left:$\frac{\psi_{23}c_{13}}{c_{23}}$] {} (X8);
	\path (X6)  edge node[label = right:$\frac{\psi_{13}(c_{13}-1)}{c_{13}}$] {} (X12);
	\end{tikzpicture}
\caption[Computation trees produced by Algorithm \ref{alg:pairs}.]{Construction of the computation tree rooted at node $x_1$ at time $t = 2$ produced by Algorithm \ref{alg:pairs} for the factor graph in Figure
\ref{factfig}.  Self-potentials are adjacent to the variable node to which they correspond. One can check that setting $c_{ij} = 1$ for all $(i,j)\in E$ reduces the above computation tree to that of Figure \ref{comptreefig}. }
\label{newcomptreefig}
\end{figure}
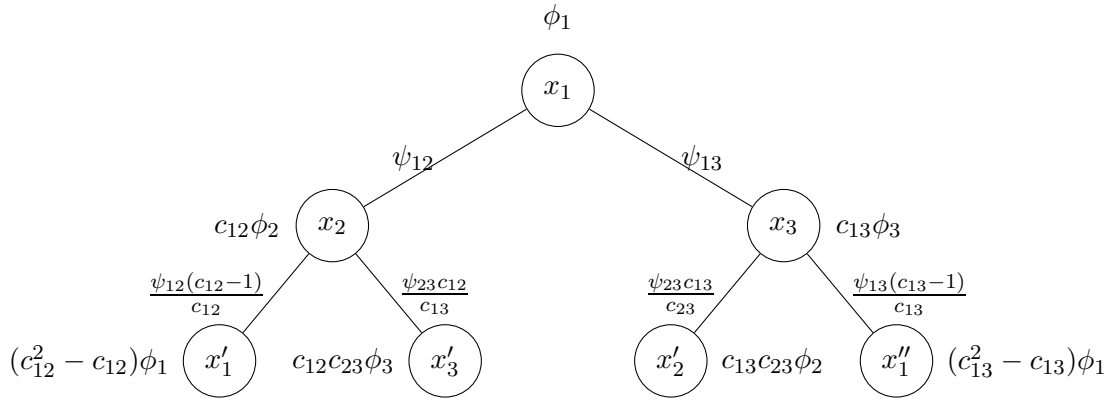

\label{newcompsec}
The computation trees produced by Algorithm \ref{alg:pairs} are
different from their predecessors.  Again, the computation tree captures the
messages that would need to be passed in order to compute $\tau^t_i(x_i)$. 
However, the messages that are passed in the new algorithm are multiplied by a
non-zero constant.  As a result, the potential at a node $u$ in the computation
tree corresponds to some potential in the original graph multiplied by a
constant that depends on all of the nodes above $u$ in the computation tree.  We
summarize the changes as follows:
\begin{enumerate}
\item The message passed from $i$ to $j$ may now depend on the message from
$j$ to $i$ at the previous time step.  As such, we now form the time $t+1$
computation tree from the time $t$ computation tree by taking any leaf $u$,
which is a copy of node $v$ in the factor graph, of the time $t$ computation
tree, creating a new node for every $w\in \partial v$, and connecting $u$ to
these new nodes.  As a result, the new computation tree rooted at node $u$ of
depth $t$ contains at least all of the non-backtracking walks of length $t$ in
the factor graph starting from $u$ and, at most, all walks of length $t$ in the
factor graph starting at $u$.
\item The messages are weighted by the elements of $c$.  This changes the
potentials at the nodes in the computation tree.  For example, suppose the
computation tree was rooted at variable node $i$ and that $\tau_i$ depends on the
message from $j$ to $i$.  Because $m_{ji}$ is multiplied
by $c_{ij}$ in $\tau_i$, every potential along this branch of the computation
tree is multiplied by $c_{ij}$.  To make this concrete, we can associate a
weight to every edge of the computation tree that corresponds to the constant
that multiplies the message passed across that edge.  To compute the new
potential at a variable node $i$ in the computation tree, we now need to
multiply the corresponding potential $\phi_i$ by each of the weights
corresponding to the edges that appear along the path from $i$ to the root of
the computation tree.  An analogous process can be used to compute the
potentials on each of the edges.  The computation tree produced by Algorithm \ref{alg:pairs} at time $t = 2$ for the factor graph in Figure \ref{factfig}
is pictured in Figure \ref{newcomptreefig}.  Compare this with computation tree
produced by the standard min-sum algorithm in Figure \ref{comptreefig}.
\end{enumerate}

If we make these adjustments, then the belief, $\tau_i^t(x_i)$, at node $i$ at time $t$ is given by the min-marginal at the root of $T_i(t)$.  In this way, the beliefs correspond to marginals at the
root of these computation trees.  
\subsection{Quadratic Minimization}

We now address the quadratic minimization problem in the context of the reweighted min-sum algorithm.  Recall that given a matrix $\Gamma$ the quadratic minimization problem is to find the vector $x$ that minimizes $f(x) = \frac{1}{2}x^T\Gamma x - h^Tx$.  Without loss of generality, we can assume that the matrix $\Gamma$ is symmetric as the quadratic function $\frac{1}{2}x^T\Gamma x - h^Tx$ is equivalent to $\frac{1}{2}x^T\Big[ \frac{1}{2}(\Gamma + \Gamma^T)\Big]x - h^Tx$ for any $\Gamma\in \mathbb{R}^{n\times n}$:
\begin{eqnarray*}
f(x) & = & \frac{1}{2}x^T\Big[ \frac{1}{2}(\Gamma + \Gamma^T) + \frac{1}{2}(\Gamma - \Gamma^T) \Big]x - h^Tx\\
& = & \frac{1}{2}x^T\Big[ \frac{1}{2}(\Gamma + \Gamma^T)\Big]x +\frac{1}{2}x^T\Big[ \frac{1}{2}(\Gamma - \Gamma^T)\Big]x - h^Tx\\
& = & \frac{1}{2}x^T\Big[ \frac{1}{2}(\Gamma + \Gamma^T)\Big]x - h^Tx.
\end{eqnarray*}

Every quadratic function admits a pairwise factorization
\begin{eqnarray*}
f(x_1,...,x_n) & = & \frac{1}{2}x^T\Gamma x - h^Tx\\
& = & \sum_i [\frac{1}{2}\Gamma_{ii}x_i^2 - h_ix_i]+ \sum_{i>j} \Gamma_{ij}x_ix_j,
\end{eqnarray*}
where $\Gamma\in \mathbb{R}^{n\times n}$ is a symmetric matrix. We note that we will abusively write $\min$ in the reweighted update equations even though the appropriate notion of minimization for the real numbers is $\inf$.   

We can explicitly compute the minimization required by the reweighted min-sum algorithm at each time step: the synchronous message update $m_{i\rightarrow j}^t(x_j)$ can be parameterized as a quadratic function of the form $\frac{1}{2}a_{i\rightarrow j}^t x_j^2 + b_{i\rightarrow j}^tx_j$.  If we define
\[A^t_{i\setminus j} \triangleq \Big[\Gamma_{ii} + \sum_{k\in \partial i} c_{ki}\cdot a_{k\rightarrow i}^{t-1}\Big] - a_{j\rightarrow i}^{t-1}\]
and
\[B^t_{i\setminus j} \triangleq \Big[h_i - \sum_{k\in \partial i} c_{ki}\cdot b_{k\rightarrow i}^{t-1}\Big] - b^{t-1}_{j\rightarrow i},\]
then the updates at time $t$ are given by
\begin{align*}
a_{i\rightarrow j}^t  \coloneqq & \frac{-\Big(\frac{\Gamma_{ij}}{c_{ij}}\Big)^2}{A^t_{i\setminus j}}\\
b_{i\rightarrow j}^t  \coloneqq & \frac{B^t_{i\setminus j} \frac{\Gamma_{ij}}{c_{ij}}}{A^t_{i\setminus j}}\hspace{.1cm}
\end{align*}
These updates are only valid when $A_{i\setminus j} > 0$.  If this is not the case, then the minimization given in Algorithm \ref{alg:pairs} is not bounded from below, and we set $a_{i\rightarrow j}^t = -\infty$.  For the initial messages, we set $a_{i\rightarrow j}^0 = b_{ij}^0 = 0$.  A similar analysis holds for the asynchronous updates.

Suppose that the beliefs generated from a fixed point of Algorithm \ref{alg:pairs} are locally decodable to $x^*$.  One can show that the gradient of $f$ at $x^*$ is always equal to zero.  If the gradient of $f$ at $x^*$ is zero and $\Gamma$ is positive definite, then $x^*$ must be a global minimum of $f$.  In other words, the min-sum algorithm always computes the correct minimizing assignment if it converges to locally decodable beliefs.  This result was proven for the GaBP algorithm in \citet{weissgauss} and the tree-reweighted algorithm in \citet{waingauss}.

\begin{theorem}
If Algorithm \ref{alg:pairs} converges to a collection of beliefs, $\tau$, that are locally decodable to $x^*$ for a quadratic function $f$, then $x^*$ is a local minimum of $f$.\label{thm:corr}
\end{theorem}
\begin{proof}
For completeness, we sketch the proof.  By Lemmas \ref{lem:admis} and \ref{lem:consis}, we have that, 
\[\min_{x_j} \tau_{ij}(x_i,x_j) = \kappa + \tau_i(x_i)\]
for all $(i,j)\in G$ and
\[f(x_1,\ldots,x_{|V|}) = \kappa + \sum_{i\in V} \tau_i(x_i) + \Big[\sum_{(i,j)\in E} \tau_{ij}(x_i,x_j) - \tau_i(x_i) - \tau_j(x_j)\Big].\]

If $\tau$ is locally decodable to $x^*$, then for each $i\in V$, $\tau_i(x_i)$ must be a positive definite quadratic function that is minimized at $x_i^*$.  Applying Lemma \ref{lem:consis}, we have that for each $(i,j)\in E$, $\tau_{ij}$ is also a positive definite quadratic function and $\tau_{ij}$ is minimized at $(x^*_i, x^*_j)$.  For each $i\in V$,
\[\frac{d}{dx_i}f(x_1,\ldots,x_{|V|}) = \frac{d}{dx_i}\tau_i(x_i) + \Big[\sum_{j\in\partial i} \frac{d}{dx_i}\tau_{ij}(x_i,x_j) - \frac{d}{dx_i}\tau_i(x_i)\Big].\]
By the above arguments, for each $i\in V$, $\frac{d}{dx_i}\tau_{i}(x_i)\Big|_{x^*} = 0$ and $\frac{d}{dx_i}\tau_{ij}(x_i,x_j)\Big|_{x^*} = 0$ for all $j\in\partial i$.  As a result, we must have $\nabla f(x^*_1,\ldots,x^*_{|V|}) = 0$.  If $\Gamma$ is positive semidefinite, then $f$ is convex, and $x^*$ must be a global minimum of $f$.
\end{proof}

As a consequence of Theorem \ref{thm:corr}, even if $\Gamma$ not positive definite, if some fixed point of the reweighted algorithm is locally decodable to a vector $x^*$ then, $x^*$ solves the system $\Gamma x = h$.

Recall that several authors have provided sufficient conditions for the convergence of GaBP:  \citet{weissgauss} demonstrated that GaBP converges in the case that the covariance matrix is diagonally dominant, \citet{malioutov} proved that the GaBP algorithm converges when the covariance matrix is walk-summable. \citet{quadciamac,convexciamac} showed that scaled diagonal dominance was a sufficient condition for convergence and also characterized the rate of convergence via a computation tree analysis.  These properties of matrices will be important in the sequel. 

\begin{definition}
$\Gamma\in \mathbb{R}^{n\times n}$ is \textbf{scaled diagonally dominant} if $\exists w>0 \in \mathbb{R}^{n}$ such that $|\Gamma_{ii}|w_i > \sum_{j\neq i} |\Gamma_{ij}|w_j$.
\end{definition}

\begin{definition}
$\Gamma\in \mathbb{R}^{n\times n}$ is \textbf{walk-summable} if the spectral radius $\varrho(|I - D^{-1/2}\Gamma D^{-1/2}|) < 1$.  Here, $D^{-1/2}$ is the diagonal matrix such that $D^{-1/2}_{ii} = \frac{1}{\sqrt{\Gamma_{ii}}}$, and $|A|$ denotes the matrix obtained from the matrix $A$ by taking the absolute value of each entry of $A$.
\end{definition}

\section{Graph Covers}
\label{sec:gc}
In this section, we will explore graph covers and their relationship to iterative message-passing algorithms for the quadratic minimization problem.  Before addressing the quadratic minimization problem specifically, we will first make a few observations about general pairwise graphical models.  The greatest strength of the above message-passing algorithms, their reliance on only local information, can also be a weakness:  local message-passing algorithms are incapable of distinguishing two graphs that have the same local structure.  To make this precise, we will need the notion of graph covers.

\begin{definition}
A graph $H$ \textbf{covers} a graph $G$ if there exists a graph homomorphism $\pi:
H \rightarrow G$ such that $h$ is an isomorphism on neighborhoods (i.e., for all vertices $i\in H$, $\partial i$ is mapped bijectively onto $\partial \pi(i)$).  If $\pi(i) = j$, then we say that $i\in H$ is a copy
of $j\in G$.  Further, $H$ is a $k$-cover of $G$ if every vertex of $G$ has
exactly $k$ copies in $H$.
\end{definition}

\begin{figure}
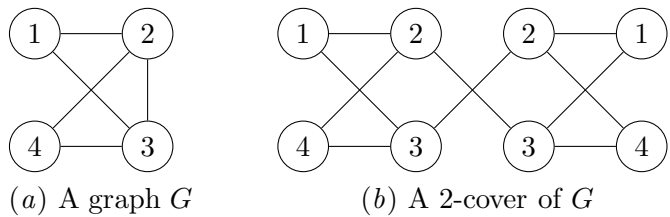

\centering
 \subfigure[A graph $G$]{\includeteximage[]{fig_1.txt}  }\hspace{1cm}
 \subfigure[A 2-cover of $G$]{\includeteximage[]{fig_2.txt}}                 
\caption[An example of a graph cover.]{An example of a graph cover.  Nodes in the cover are labeled by the node that they are a copy of in $G$.}
\label{fig:gcoverex}
\end{figure}

Graph covers, in the context of graphical models, were originally studied in relation to local message-passing algorithms for coding problems \citep{vontobel}.  Graph covers may be connected (i.e., there is a path between every pair of vertices) or disconnected.  However, when a graph cover is disconnected, all of the connected components of the cover must themselves be covers of the original graph.  For a simple example of a connected graph cover, see Figure \ref{fig:gcoverex}.  

Every finite cover of a connected graph is a $k$-cover for some integer $k$.  For every base graph $G$, there exists a graph, possibly infinite, which covers all finite, connected covers of the base graph.  This graph is known as the universal cover. 

To any finite cover, $H$, of a factor graph $G$ we can associate a collection of potentials derived from the base graph; the
potential at node $i\in H$ is equal to the potential at node $h(i) \in G$.  Together, these potential functions define a new objective function for the factor graph $H$.  In the sequel, we will use superscripts to specify that a particular object is over the factor graph $H$.  For example, we will denote the objective function corresponding to a factor graph $H$ as $f^H$, and we will write $f^G$ for the objective function $f$.

Local message-passing algorithms such as the reweighted min-sum  algorithm are incapable of distinguishing the two factor graphs $H$ and $G$ given that the initial messages to and from each node in $H$ are identical
to the nodes that they cover in $G$: for every node $i\in G$ the
messages received and sent by this node at time $t$ are exactly the same as the
messages sent and received at time $t$ by any copy of $i$ in $H$.  As a result,
if we use a local message-passing algorithm to deduce an assignment for $i$, then the algorithm
run on the graph $H$ must deduce the same assignment for each copy of $i$.

Now, consider an objective function $f$ that factors over the graph $G$.  For any finite cover $H$ of $G$ with covering homomorphism $h:H\rightarrow G$, we can "lift" any vector of beliefs, $\tau^G$, from $G$ to $H$ by defining a new vector of beliefs, $\tau^H$, such that:
\begin{itemize}
\item For all variable nodes $i\in H$, $\tau^H_i = \tau^G_{\pi(i)}$.
\item For all edges $(i,j)\in H$, $\tau^H_{ij} = \tau^G_{\pi(i)\pi(j)}$.
\end{itemize}
Analogously, we can lift any assignment $x^G$ to an assignment $x^H$ by setting $x^H_i = x^G_{\pi(i)}$.

\subsection{Graph Covers and Quadratic Minimization}
\label{sec:gcovers}
Let $G$ be the pairwise factor graph for the objective function $f^G(x_1,..., x_n) = \frac{1}{2}x^T\Gamma x - h^Tx$ whose edges correspond to the nonzero entries of $\Gamma$. Let $H$ be a $k$-cover of $G$ with corresponding objective function $f^H(x_{11},...,x_{1k},...x_{nk}) = \frac{1}{2}x^T\widetilde{\Gamma} x - \widetilde{h}^Tx$.  Without loss of generality we can assume that $\widetilde{\Gamma}$ and $\widetilde{h}$ take the following form:
\begin{eqnarray}
\label{eq:perm}
\widetilde{\Gamma} & = & \begin{pmatrix}
\Gamma_{11}P_{11} & \cdots & \Gamma_{1n}P_{1n}\nonumber\\
\vdots & \ddots & \vdots\\
\Gamma_{n1}P_{n1} & \cdots & \Gamma_{nn}P_{nn}
\end{pmatrix}\\
\widetilde{h}_i & = & h_{\left\lceil i/k\right\rceil},
\end{eqnarray}
where $P_{ij}$ is a $k\times k$ permutation matrix for all $i \neq j$ and $P_{ii}$ is the $k\times k$ identity matrix for all $i$.  If $\widetilde{\Gamma}$ is dervied from $\Gamma$ in this way, then we will say that $\widetilde{\Gamma}$ covers $\Gamma$.

For the quadratic minimization problem, factor graphs and their covers share many of the same properties.  Most notably, we can transform critical points of covers to critical points of the original problem.  Let $H$ and $G$ be as above, and let $\pi$ be the graph homomorphism from $H$ to $G$. For $x\in\mathbb{R}^{|V_G|}$, define $\text{lift}_H(x)\in\mathbb{R}^{|V_H|}$ such that \[\text{lift}_H(x)_i = x_{\pi(i)}\] for all $i\in H$.  Similarly, for each $y\in\mathbb{R}^{|V_H|}$, define $\text{proj}_G(y)\in\mathbb{R}^{|V_G|}$ such that \[\text{proj}(y)_i = \sum_{k\in H : h(k) = i} \frac{y_k}{|\{j\in H : h(j) = i\}|}\] for all $i\in G$.  With these definitions, we have the following lemma.

\begin{lemma}
\label{crit}
If $\widetilde{\Gamma} y = \widetilde{h}$ for $y\in \mathbb{R}^{|V_H|}$, then $\Gamma\cdot \text{proj}_G(y) = h$.  Conversely, if $\Gamma x = h$ for $y\in \mathbb{R}^{|V_G|}$, then $\widetilde{\Gamma}\cdot \text{lift}_H(x) = \widetilde{h}$.
\end{lemma}

Notice that these solutions correspond to critical points of the cover and the original problem.  Similarly, we can transform eigenvectors of covers to either eigenvectors of the original problem or the zero vector.

\begin{lemma}
\label{eigen}
Fix $\lambda\in\mathbb{R}$.  If $\widetilde{\Gamma} y = \lambda y$, then either $\Gamma\cdot \text{proj}_G(y) = \lambda \text{proj}_G(y)$ or $\Gamma\cdot \text{proj}_G(y) = 0$.  Conversely, if $\Gamma x = \lambda x$, then $\widetilde{\Gamma}\cdot\text{lift}_H(x) = \lambda \text{lift}_H(x)$.
\end{lemma}

These lemmas demonstrate that we can average critical points and eigenvectors of covers to obtain critical points and eigenvectors (or the zero vector) of the original problem, and we can lift critical points and eigenvectors of the original problem in order to obtain critical points and eigenvectors of covers. 

\begin{figure}
\begin{center}
\begin{tabular}{cc}
$\Gamma  =  \begin{pmatrix}
1 & .6 & .6\\
.6 & 1 & .6\\
.6 & .6 & 1
\end{pmatrix}$ & $\widetilde{\Gamma} =  \begin{pmatrix}
1 & 0 & .6 & 0 & 0 & .6\\
0 & 1 & 0 & .6 & .6 & 0\\
.6 & 0 & 1 & 0 & .6 & 0\\
0 & .6 & 0 & 1 & 0 & .6\\
0 & .6 & .6 & 0 & 1 & 0\\
.6 & 0 & 0 & .6 & 0 & 1
\end{pmatrix}$
\end{tabular}
\end{center}
\caption[A positive definite matrix that is covered by a matrix with negative eigenvalues.]{An example of a positive definite matrix, $\Gamma$, which possesses a 2-cover, $\widetilde{\Gamma}$, that has negative eigenvalues.}
\label{fig:quadcover}
\end{figure}

Unfortunately, even though the critical points of $G$ and its covers must correspond via Lemma \ref{crit}, the corresponding minimization problems may not have the same solution.  The example in Figure \ref{fig:quadcover} illustrates that there exist positive definite matrices that are covered by matrices that are not positive definite.  This observation seems to be problematic for the convergence of iterative message-passing schemes.  Specifically, the fixed points of the reweighted algorithm on the base graph are also fixed points of the reweighted algorithm on any graph cover.  As such, the reweighted algorithm may not converge to the correct minimizing assignment when the matrix corresponding to some cover of $G$ is not positive definite.  Consequently, we will first consider the special case in which $\Gamma$ and all of its covers are positive definite.  We can exactly characterize the matrices for which this property holds.

\begin{theorem}
\label{2cover}
Let $\Gamma$ be a symmetric matrix with positive diagonal. The following are equivalent:
\begin{enumerate}
\item $\Gamma$ is walk-summable.
\item $\Gamma$ is scaled diagonally dominant.
\item All covers of $\Gamma$ are positive definite.
\item All 2-covers of $\Gamma$ are positive definite.
\end{enumerate}
\end{theorem}
\begin{proof}
The two non-trivial implications in the proof ($4 \Rightarrow 1$ and $1 \Rightarrow 2$) make use of the Perron-Frobenius theorem.  For the complete details, see Appendix \ref{ap:2cover}.
\end{proof}

This theorem has several important consequences.  First, it provides us with a combinatorial characterization of scaled diagonal dominance and walk-summability.  Second, it provides an intuitive explanation for why these conditions should be sufficient for the convergence of local message-passing algorithms.  Indeed,  walk-summability and scaled diagonal dominance were independently shown to be sufficient conditions for the convergence of the min-sum algorithm for positive definite matrices \citep{malioutov, convexciamac}.  Most importantly, we can use the theorem to conclude that MPLP, tree-reweighted max-product, and other message-passing algorithms that guarantee the correctness of locally decodable beliefs cannot converge to the correct solution when $\Gamma$ is positive definite but not walk-summable.  To see this, note that, for these algorithms, if the beliefs are locally decodable to a vector $x^*$, then $x^*$ must minimize the objective function.  As we saw earlier, any collection of locally decodable beliefs on the base graph can be lifted to locally decodable beliefs on any graph cover.  In other words, the lift of $x^*$ to each graph cover must be a global minimum on that cover.  However, there exists at least one cover with no global minimum.  As a result, these algorithms cannot converge to locally decodable beliefs.   For more details about these types of message-passing algorithms, we refer the reader to \citet{MPLP},\citet{waintree}, and \citet{sontag}.

Contrast this analysis with Theorem \ref{thm:corr}.  In general, the reweighted message-passing algorithm only guarantees that $x^*$ is a local optimum whenever the objective function is not positive semidefinite, but there exist simple choices for the reweighting parameters that guarantee correctness over all covers.  As an example, if $c_{ij} \leq \frac{1}{\max_{i\in V} |\partial i|}$ for all $(i,j)\in E$, then one can show that the reweighted algorithm cannot converge to locally decodable beliefs. The traditional choice of parameters for the TRMP algorithm where each $c_{ij}$ corresponds to an edge appearance probability provides another example.  As such, in order to produce convergent message-passing schemes for the quadratic minimization problem, we will need to study choices of the parameters that do not guarantee correctness over all graph covers.

\section{Convergence Properties of Reweighted Message-Passing Algorithms}
Recall that GaBP algorithm can converge to the correct minimizer of the objective function even if the original matrix is not scaled diagonally dominant.  The most significant problem when the original matrix is positive definite but not scaled diagonally dominant is that the computation trees may eventually possess negative eigenvalues due to the existence of some 2-cover with at least one non-positive eigenvalue.  If this happens, then some of the beliefs will not be bounded from below, and the corresponding estimate will be negative infinity.  This is, of course, the correct answer on some 2-cover of the problem, but it is not the correct solution to the minimization problem of interest.  

Our goal in this section is to understand how the choice of the parameters affects the convergence of the reweighted algorithm.

\label{sec:convmsg}
\subsection{Convergence of the Variances}
\label{sec:vars}
  In this section, we will provide conditions on the choice of the parameter vector such that all of the computation trees produced by the reweighted algorithm remain positive definite throughout the course of the algorithm.

Positive definiteness of the computation trees corresponds to the convexity of the beliefs, and the convexity of the belief, $\tau^t_i$, is determined only by the vector $a^t$.  As such, we begin by studying the sequence $a^0, a^1,...$ where $a^0$ is the zero vector (based on our initialization).  We will consider two different choices for the parameter vector: one in which $c_{ij} \geq 1$ for all $i$ and $j$ and one in which $c_{ij} \leq 0$ for all $i$ and $j$.  

\subsubsection{Positive Parameters}
\begin{lemma}
\label{mono}
If $c_{ij} \geq 1$ for all $i$ and $j$, then for all $t > 0$, $a^t_{i\rightarrow j} \leq a^{t-1}_{i\rightarrow j} \leq 0$ for each $i$ and $j$.
\end{lemma}
\begin{proof}
This result follows by induction on $t$.  First, suppose that $c_{ij} \geq 1$.  If the update is not valid, then $a_{i\rightarrow j}^t = -\infty$ which trivially satisfies the inequality. Otherwise, we have:
\begin{eqnarray*}
a_{ij}^t & = & \frac{-\Big(\frac{\Gamma_{ij}}{c_{ij}}\Big)^2}{\Gamma_{ii} + \sum_{k\in \partial i \setminus j} c_{ki}a_{k\rightarrow i}^{t-1} + (c_{ji}-1)a_{j\rightarrow i}^{t-1}}\hspace{.5cm}\\
& \leq & \frac{\Big(\frac{\Gamma_{ij}}{c_{ij}}\Big)^2}{\Gamma_{ii} + \sum_{k\in \partial i \setminus j} c_{ki}a_{k\rightarrow i}^{t-2} + (c_{ji}-1)a_{j\rightarrow i}^{t-2}}\\
& = & a_{i\rightarrow j}^{t-1},
\end{eqnarray*}
where the inequality follows from the observation that $\Gamma_{ii} + \sum_{k\in \partial i \setminus j} c_{ki}a_{k\rightarrow i}^{t-1} + (c_{ji}-1)a_{j\rightarrow i}^{t-1} > 0$ and the induction hypothesis.
\end{proof}

If we consider only the vector $a^t$, then the algorithm may exhibit a weaker form of convergence:
\begin{lemma}
If $c_{ij} \geq 1$ for all $i$ and $j$ and all of the computation trees are positive definite, then the sequence $a_{i\rightarrow j}^0, a_{i\rightarrow j}^1,...$ converges.
\end{lemma}
\begin{proof}
Suppose $c_{ij} \geq 1$.  By Lemma \ref{mono}, the $a_{i\rightarrow j}^t$ are monotonically decreasing.  Because all of the computation trees are positive definite, we must have that for each $i$, $\Gamma_{ii} + \sum_{k\in \partial i \setminus j} c_{ki}a_{k\rightarrow i}^{t-1} + c_{ji}a_{j\rightarrow i}^{t-1} > 0$.  Therefore, for all $(i,j)\in E$, $a^t_{i\rightarrow j} \geq -\frac{\Gamma_{ii}}{c_{ij}}$, and the sequence $a_{i\rightarrow j}^0,a_{i\rightarrow j}^1,...$ is monotonically decreasing and bounded from below.  This implies that the sequence converges.  \\
\end{proof}

Because the estimates of the variances only depend on the vector $a^t$, if the $a_{i\rightarrow j}^t$ converge, then the estimates of the variances also converge.  Therefore, requiring all of the computation trees to be positive definite is a sufficient condition for convergence of the variances. Note, however, that the estimates of the means which correspond to the sequence $b_{i\rightarrow j}^t$ need not converge even if all of the computation trees are positive definite (see Figure \ref{fig:varsnomeans}).

\begin{figure}
\begin{center}
$\begin{pmatrix}
                         1     &              0.39866      &            -0.39866        &          -0.39866\\
                   0.39866      &                   1      &            -0.39866        &                 0\\
                  -0.39866      &            -0.39866      &                   1        &          -0.39866\\
                  -0.39866      &                   0      &            -0.39866        &                 1
                  \end{pmatrix}$
\end{center}
\caption[A positive definite matrix for which the variances in the standard min-sum algorithm converge but the means do not.]{A positive definite matrix for which the variances in the min-sum algorithm converge but the means do not. \citep{malthesis}}
\label{fig:varsnomeans}
\end{figure}

Our strategy will be to ensure that all of the computation trees are positive definite by leveraging the choice of parameters, $c_{ij}$.  Specifically, we want to use these parameters to weight the diagonal elements of the computation tree much more than the off-diagonal elements in order to force the computation trees to be positive definite.  If we can show that there is a choice of each $c_{ij} = c_{ji}$ that will cause all of the computation trees to be positive definite, then Algorithm \ref{alg:pairs} should behave almost as if the original matrix were scaled diagonally dominant.  There always exists a choice of the vector $c$ that achieves this.

\begin{theorem}
\label{posdefcomp}
For any symmetric matrix $\Gamma$ with strictly positive diagonal, $\exists r \geq 1$ and an $\epsilon > 0$ such that the eigenvalues of the computation trees are bounded from below by $\epsilon$ when generated by Algorithm \ref{alg:pairs} with $c_{ij} = r$ for all $i$ and $j$.
\end{theorem}
\begin{proof}
The proof of this theorem exploits the Ger\v{s}gorin disc theorem in order to show that there exists a choice of $r$ such that each computation tree is scaled diagonally dominant.  The complete proof can be found in Appendix \ref{ap:posdefcomp}.
\end{proof}

\subsubsection{Negative Parameters}
For the case in which $c_{ij} < 0$ for all $i$ and $j$, we also have that the computation trees are always positive definite when the initial messages are uniformly equal to zero as characterized by the following lemmas.

\begin{lemma}
\label{neg}
If $c_{ij} < 0$ for all $i$ and $j$, then for all $t > 0$, $a^t_{i\rightarrow j} \leq 0$.
\end{lemma}
\begin{proof}
This result follows by induction on $t$.  First, suppose that $c_{ij} \geq 0$.  If the update is not valid, then $a_{i\rightarrow j}^t = -\infty$ which trivially satisfies the inequality. Otherwise, we have
\begin{eqnarray*}
a_{i\rightarrow j}^t & = & \frac{-\Big(\frac{\Gamma_{ij}}{c_{ij}}\Big)^2}{\Gamma_{ii} + \sum_{k\in \partial i \setminus j} c_{ki}a_{k\rightarrow i}^{t-1} + (c_{ji}-1)a_{j\rightarrow i}^{t-1}}\hspace{.5cm}\\
& \leq & 0,
\end{eqnarray*}
where the inequality follows from the induction hypothesis.
\end{proof}

\begin{lemma}
For any symmetric matrix $\Gamma$ with strictly positive diagonal, if $c_{ij} < 0$ for all $i$ and $j$, then all of the computation trees are positive definite.
\end{lemma}
\begin{proof}
The computation trees are all positive definite if and only if $\Gamma_{ii} + \sum_{k\in \partial i} c_{ki}a_{k\rightarrow i}^{t} > 0$ for all $t$.  By Lemma \ref{neg}, $a^t_{i\rightarrow j} \leq 0$ for all $t$, and as result, $\Gamma_{ii} + \sum_{k\in \partial i} c_{ki}a_{k\rightarrow i}^{t} \geq \Gamma_{ii} > 0$ for all $t$.
\end{proof}

As in the case when $c_{ij} \geq 1$ for all $(i,j)\in E$, the eigenvalues on each computation tree are again bounded way from zero, but the $a_{i\rightarrow j}^t$ no longer form a monotonic decreasing sequence when $c_{ij} < 0$ for all $(i,j)\in E$.  If all of the computation trees remain positive definite in the limit, then the beliefs will all be positive definite upon convergence.  If the estimates for the means converge as well, then the converged beliefs must be locally decodable to the correct minimizing assignment.  Notice that none of the above arguments for the variances require $\Gamma$ to be positive definite.  Indeed, we have already seen an example of a matrix with a strictly positive diagonal and negative eigenvalues (see the matrix in Figure \ref{fig:quadcover}) such that variance estimates converge.

\subsection{Synchronous Versus Asynchronous Updates}
\label{sec:sync}
The synchronous message-passing updates described in Algorithm \ref{alg:pairs} enforce a particular ordering on the updates performed at each time step.  We can construct an asynchronous version of Algorithm \ref{alg:pairs} by allowing some arbitrary ordering of message updates.  The resulting asynchronous algorithm is given by Algorithm \ref{paira}.  Because each asynchronous computation tree is a principal submatrix of a synchronous computation tree and principal submatrices of positive definite matrices are positive defintie, we can easily check that all of the results of the previous section extend to this asynchronous algorithm as well.

Asynchronous algorithms allow for quite a bit more flexibility in the scheduling of message updates, and as we will see experimentally in Section \ref{sec:exp}, asynchronous algorithms can have better convergence properties than the corresponding synchronous algorithms.  To see why this might be the case, we will again exploit the properties of graph covers.  Specifically, we will show that these two algorithms are related via a special 2-cover of the base factor graph.

\begin{algorithm}[t]
\caption{Asynchronous Reweighted message-passing Algorithm\label{paira}}
\begin{algorithmic}[1]
\STATE Initialize the messages to some finite vector.

\STATE Choose some ordering of the variables such that each variable is updated infinitely often, and perform the following update for each variable $j$ in order

\FOR{each $i \in \partial j$}
\STATE Update the message from $i$ to $j$:
\begin{align*}
m_{i\rightarrow j}(x_j) \coloneqq & \kappa + \min_{x_i} \Big[ \frac{\psi_{ij}(x_i,x_j)}{c_{ij}} + (c_{ij} - 1)m_{j\rightarrow i}(x_i)+\phi_i(x_i) +  \sum_{k\in\partial i \setminus j} c_{ki} m_{k\rightarrow i}(x_i)\Big].
\end{align*}
\ENDFOR
\end{algorithmic}
\end{algorithm}

Every pairwise factor graph, $G = (V_G,E_G)$, admits a bipartite 2-cover, $H = (V_G\times\{1,2\},
E_H)$, called the Kronecker double cover of $G$.  We will denote copies of the variable $x_i$ in this 2-cover as $x_{i_1}$ and $x_{i_2}$.  For every edge $(i,j)\in E_G$, $(i_1, j_2)$ and $(i_2, j_1)$
belong to $E_H$.  In this way, nodes labeled with a one are only connected to
nodes labeled with a two (see Figure \ref{fig:2cover}).  Note that if $G$ is already a bipartite graph, then the Kronecker double cover of $G$ is simply two disjoint copies of $G$.

\begin{figure}
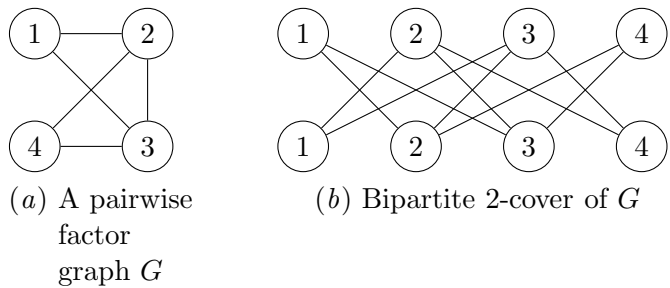

\centering
 \subfigure[A pairwise factor graph $G$]{\includeteximage[]{fig_3.txt}  }\hspace{1cm}
 \subfigure[Bipartite 2-cover of $G$]{\includeteximage[]{fig_4.txt}}                 
\caption[Kronecker double cover of a pairwise factor graph.]{The Kronecker double cover (b) of a pairwise factor graph (a).  The node labeled $i\in G$ corresponds to the variable node $x_i$.}
\label{fig:2cover}
\end{figure}

We can view the synchronous algorithm described in Algorithm \ref{alg:pairs} as a
specific asynchronous algorithm on the Kronecker double cover where we perform the asynchronous update
for every variable in the same partition on alternating iterations (see
Algorithm \ref{pairbi}).  

\begin{algorithm}[t]
\caption{Bipartite Asynchronous Algorithm\label{pairbi}}
\begin{algorithmic}[1]
\STATE Initialize the messages to some finite vector.

\STATE Iterate the following until convergence: update all of the outgoing messages from nodes
labeled one to nodes labeled two and then update all of the outgoing messages from
nodes labeled two to nodes labeled one using the asynchronous update rule:
\begin{align*}
 m_{i\rightarrow j}(x_j)  \coloneqq & \kappa + \min_{x_i} \Big[ 
\frac{\psi_{ij}(x_i,x_j)}{c_{ij}} + (c_{ij} - 1)m_{j\rightarrow i}(x_i)+\phi_i(x_i) + \sum_{k\in\partial i \setminus j} c_{ki} m_{k\rightarrow i}(x_i)\Big].
\end{align*}
\end{algorithmic}
\end{algorithm}

By construction, the message vector produced by Algorithm
\ref{pairbi} is simply a concatenation of two consecutive time steps of the
synchronous algorithm.  Specifically, for all $t\geq 1$
\begin{eqnarray*}
m_H^t = \begin{bmatrix}m_G^{2t-1}\\ m_G^{2t-2}\end{bmatrix}.  
\end{eqnarray*}

Therefore, the messages passed by Algorithm \ref{alg:pairs} are identical to those passed by an asynchronous algorithm on the Kronecker double cover.  From our earlier analysis, we know that even if $\Gamma$ is positive definite, not every cover necessarily corresponds to a convex objective function.  If the Kronecker double cover is such a ``bad'' cover, then we might expect that synchronous reweighted algorithm may not converge to the correct solution.  This reasoning is not unique to iterative message-passing algorithms.  In the next section, we will see that it can be applied to other iterative techniques for quadratic minimization.

\subsubsection{The Gauss-Seidel and Jacobi Methods}
\label{sec:other}
Because minimizing symmetric positive definite quadratic functions is equivalent to solving symmetric positive definite linear systems, well-studied algorithms such as Gaussian elimination, Cholesky decomposition, etc. can be used to compute the minimum.  In addition, many iterative algorithms have been proposed to solve the linear system $\Gamma x = h$: Gauss-Seidel iteration, Jacobi iteration, the algebraic reconstruction technique, etc.

In this section, we will show that the previous graph cover analysis can also be used to reason about the Jacobi and Gauss-Seidel algorithms (Algorithms \ref{j} and \ref{gs}).  When $\Gamma$ is symmetric positive definite, the objective function, $\frac{1}{2}x^T\Gamma x - h^Tx$, is a convex function of $x$.  Consequently, we could use a coordinate descent scheme in an attempt to minimize the objective function.  The standard cyclic coordinate descent algorithm for this problem is known as the Gauss-Seidel algorithm.  

In the same way that Algorithm \ref{alg:pairs} is a synchronous version of Algorithm \ref{paira}, the Jacobi algorithm is a synchronous version of the Gauss-Seidel algorithm.  To see this, observe that the iterates produced by the Jacobi algorithm are related to the iterates of the Gauss-Seidel algorithm on a larger problem.  Specifically, given a symmetric $\Gamma\in\mathbb{R}^{n\times n}$ and $h\in\mathbb{R}^{n}$, construct $\Gamma'\in\mathbb{R}^{2n\times 2n}$ and $h'\in\mathbb{R}^{2n}$ as follows:
\begin{eqnarray*}
h'_i & = & h_{\lceil h_i/n\rceil}\\
\Gamma' & = & \begin{bmatrix}D & M\\ M & D\end{bmatrix},
\end{eqnarray*}
where $D$ is a diagonal matrix with the same diagonal entries as $\Gamma$ and $M = \Gamma - D$.

\begin{algorithm}[t]
\caption{Jacobi Iteration\label{j}}
\begin{algorithmic}[1]
\STATE Choose an initial vector $x^0 \in \mathbb{R}^n$.

\STATE For iteration $t = 1,2,...$ set
\[x^t_j = \frac{h_j - \sum_{k} \Gamma_{jk} x^{t-1}_k}{\Gamma_{jj}}\]
for each $j\in \{1,...,n\}$.
\end{algorithmic}
\end{algorithm}

\begin{algorithm}[t]
\caption{Gauss-Seidel Iteration\label{gs}}
\begin{algorithmic}[1]
\STATE Choose an initial vector $x \in \mathbb{R}^n$.

\STATE Choose some ordering of the variables, and perform the following update
for each variable $j$, in order:
\[x_j = \frac{h_j - \sum_{k} \Gamma_{jk} x_k}{\Gamma_{jj}}.\]
\end{algorithmic}
\end{algorithm}

$\Gamma'$ is the analog of the Kronecker double cover discussed in Section \ref{sec:sync}.  Let $x^0\in\mathbb{R}^n$ be an initial vector for the Jacobi algorithm performed on the matrix $\Gamma$ and fix $y^0\in\mathbb{R}^{2n}$ such that $y^0_i = x_{1 + (i-1 \bmod {n})}$.  Further, suppose that we update the variables in the order 1,2,...,2n in the Gauss-Seidel algorithm.  If $y^t$ is the vector produced after $t$ complete cycles of the Gauss-Seidel algorithm, then $y^t = \begin{bmatrix}x^{2t-1}\\x^{2t}\end{bmatrix}$.  Also, observe that, for any $y^t$ such that $\Gamma' y^t = h'$, we must have that $\Gamma \Big[\frac{x^{2t-1} + x^{2t}}{2}\Big] = h$.  

With these two observations, any convergence result for the Gauss-Seidel algorithm can be extended to the Jacobi algorithm.  Consider the following:

\begin{theorem}
Let $\Gamma$ be a symmetric positive semidefinite matrix with a strictly positive diagonal.  The Gauss-Seidel algorithm converges to a vector $x^*$ such that $\Gamma x^* = h$ whenever such a vector exists.  
\end{theorem}
\begin{proof}
See Section 10.5.1 of \citet{byrne}.
\end{proof}
\hspace{1cm}\\*
Using our observations, we can immediately produce the following new result:
\begin{corollary}
Let $\Gamma$ be a symmetric positive semidefinite matrix with positive diagonal and let $\Gamma'$ be constructed as above.  If $\Gamma'$ is a symmetric positive semidefinite matrix and there exists an $x^*$ such that $\Gamma x^* = h$, then the sequence $\frac{x^t + x^{t-1}}{2}$ converges to $x^*$ where $x^t$ is the $t^{th}$ iterate of the Jacobi algorithm.
\end{corollary}

If $\Gamma'$ is not positive semidefinite, then the Gauss-Seidel algorithm (and by extension the Jacobi algorithm) may or may not converge when run on $\Gamma'$. 

\subsection{Convergence of the Means}
\label{sec:means}
If the variances converge, then the fixed points of the message updates for the means correspond to the solution of a particular linear system $Mb = d$.  In fact, we can show that Algorithm \ref{paira} is exactly the Gauss-Seidel algorithm for this linear system.  First, we construct the matrix $M\in\mathbb{R}^{2|E|\times 2|E|}$:
\begin{eqnarray*}
M_{ij, ij} & = & A^*_{i\setminus j} \text{ for all $i\in V$ and $j\in\partial i$}\hspace{.5cm}\\
M_{ij, ki} & = & c_{ki}\frac{\Gamma_{ij}}{c_{ij}} \text{ for all $i\in V$ and for all $j,k\in\partial i$ such that $k\neq j$}\\
M_{ij, ji} & = & (c_{ij}-1)\frac{\Gamma_{ij}}{c_{ij}} \text{ for all $i\in V$ and $j\in\partial i$}.
\end{eqnarray*}
Here, $A^*$ is constructed from the vector of converged variances, $a^*$.  All other entries of the matrix are equal to zero.  Next, we define the vector $d\in\mathbb{R}^{2|E|}$ by setting $d_{ij} = h_i\Gamma_{ij}/c_{ij}$ for all $i\in V$ and $j\in\partial i$.  

By definition, any fixed point, $b^*$, of the message update equations for the means must satisfy $Mb^*=d$.  With these definitions, Algorithm \ref{paira} is precisely the Gauss-Seidel algorithm for this matrix.  Similarly, Algorithm \ref{alg:pairs} corresponds to the Jacobi algorithm.  Unfortunately, $M$ is neither symmetric nor diagonally dominant, so the standard results for the convergence of the Gauss-Seidel algorithm do not necessarily apply to this situation.  In practice, we have observed that the asynchronous reweighted message-passing algorithm converges if each $c_{ij}$ is sufficiently large (or sufficiently negative).

\subsection{Experimental Results}
\label{sec:exp}

\begin{figure}
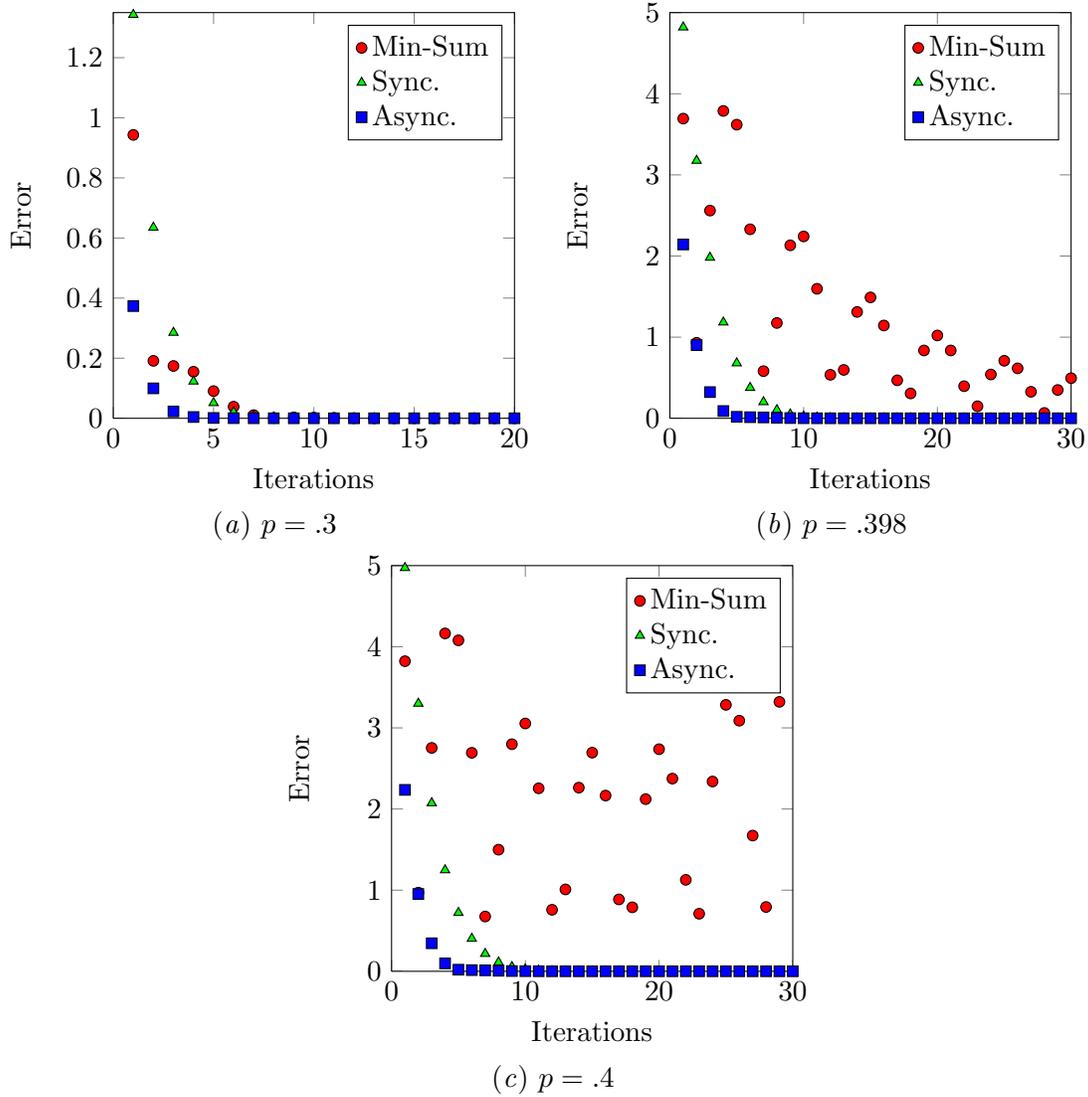

\centering
\subfigure[$p = .3$]{\label{res1}\includeteximage[]{fig_5.txt}}
\subfigure[$p = .398$]{\label{res2}\includeteximage[]{fig_6.txt}}
\subfigure[$p = .4$]{\label{res3}\includeteximage[]{fig_7.txt}}
\caption[A comparison of the performance of the min-sum algorithm and the reweighted algorithm for the quadratic minimization problem.]{The error, measured by the 2-norm, between the current mean estimate and the true mean at each step of the min-sum algorithm, the asynchronous algorithm with $c_{ij}=2$ for all $i\neq j$, and the synchronous algorithm with $c_{ij}=2$ for all $i\neq j$ for the matrix in \eqref{eqn:4chord}.  Notice that all of the algorithms have a similar performance when $p$ is chosen such that the matrix is scaled diagonally dominant.  When the matrix is not scaled diagonally dominant, the min-sum algorithm converges more slowly or does not converge at all.}
\label{fig:exp_chord}
\end{figure}

Even simple experiments demonstrate the advantages of the reweighted message-passing algorithm compared to the typical min-sum algorithm. Throughout this section, we will assume that $h$ is chosen to be the vector of all ones. Let $\Gamma$ be the following matrix.
\begin{eqnarray}
\begin{pmatrix}
                         1     &              p      &            -p       &          -p\\
                   p      &                   1      &            -p        &                 0\\
                  -p      &            -p    &                   1        &          -p\\
                  -p     &                   0      &            -p        &                 1
                  \end{pmatrix}\label{eqn:4chord}
\end{eqnarray}
The standard min-sum algorithm converges to the correct solution for \hspace{1pt}$0 \leq p < .39865$ \citep{malioutov}.  Figure \ref{fig:exp_chord} illustrates the behavior of the min-sum algorithm, the asynchronous reweighted message-passing algorithm with $c_{ij}=2$ for all $i\neq j$, and the synchronous algorithm with $c_{ij}=2$ for all $i\neq j$ for different choices of the constant $p$.  Each iteration of the asynchronous algorithm consists of cyclically updating all messages.  In the examples in Figure \ref{fig:exp_chord}, the synchronous and asynchronous algorithm always converge rapidly to the correct mean while the min-sum algorithm converges slowly or not at all as $p$ approaches $.5$.

While this is a simple graph, the behavior of the algorithm for different choices of the vector $c$ is already apparent.  If we set $c_{ij}=3$ for all $i\neq j$, then empirically, both the synchronous and asynchronous algorithms converge for all $p \in (-.5, .5)$, which is the entire positive definite region for this matrix.  However, different choices of the parameter vector can greatly increase or decrease the number of iterations required for convergence.  Figure \ref{fig:c} illustrates the iterations to convergence for the reweighted algorithms at $p = .4$ versus $c$.  

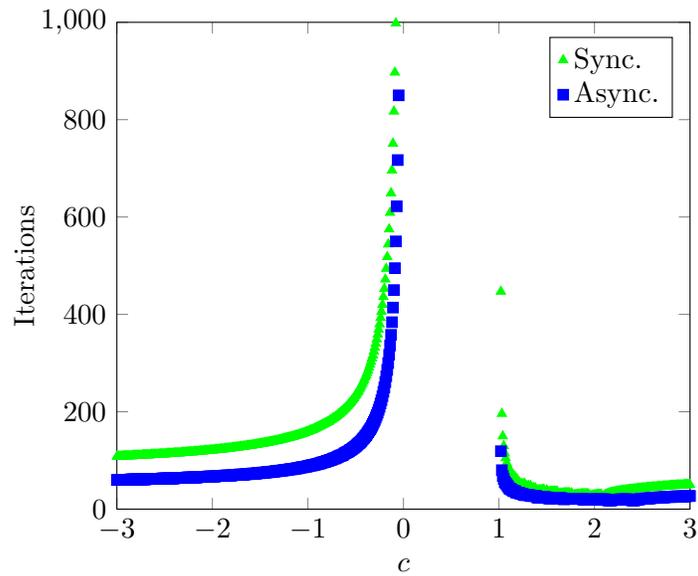
\begin{figure}
\centering

\begin{tikzpicture}[only marks]
\begin{axis}[scale only axis,
width=0.5\textwidth,
height=0.3\textheight,
xmin=-3, xmax=3,
ymin=0, ymax=1000, legend entries={Sync., Async.},
legend style={nodes=right},
legend pos= north east,
xlabel = $c$,
ylabel = Iterations,
xticklabel style={/pgf/number format/.cd,fixed,precision=5}
]
		\addplot[mark=triangle*, mark options={color=green,fill=green}] 
			table [x=x,y=z] {c_plot.csv};
		\addplot[mark=square*, mark options={color=blue,fill=blue}] 
			table [x=x,y=y] {c_plot.csv};
			
\end{axis}
	\end{tikzpicture}
	\caption{The number of iterations needed to reduce the error of the mean estimates below $10^{-6}$ using the reweighted algorithms as a function of $c$  for the matrix in \eqref{eqn:4chord} with $p = .4$.  The gap in the plot is predicted by the arguments at the end of Section \ref{sec:gcovers}.}
	\label{fig:c}
\end{figure}

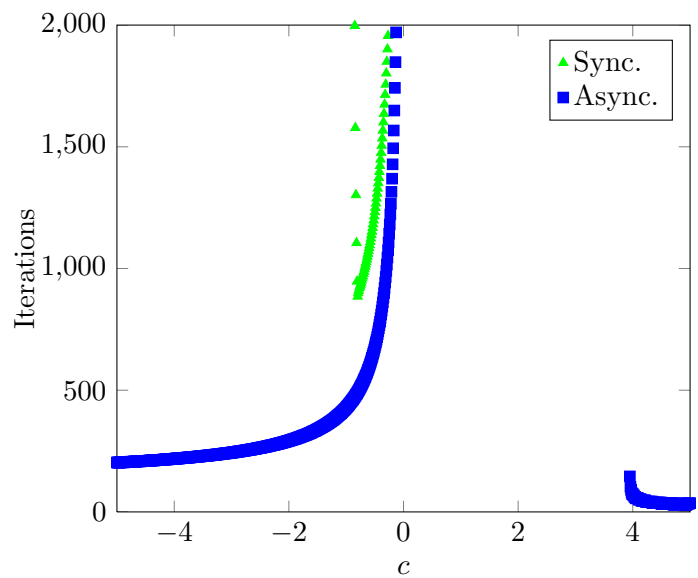
\begin{figure}
\centering

\begin{tikzpicture}[only marks]
\begin{axis}[scale only axis,
width=0.5\textwidth,
height=0.3\textheight,
xmin=-5, xmax=5,
ymin=0, ymax=2000, legend entries={Sync., Async.},
legend style={nodes=right},
legend pos= north east,
xlabel = $c$,
ylabel = Iterations,
xticklabel style={/pgf/number format/.cd,fixed,precision=5}
]
		\addplot[mark=triangle*, mark options={color=green,fill=green}] 
			table [x=x,y=z] {rnd_plot.csv};
		\addplot[mark=square*, mark options={color=blue,fill=blue}] 
			table [x=x,y=y] {rnd_plot.csv};
			
\end{axis}
	\end{tikzpicture}
	\caption{The number of iterations needed to reduce the error of the mean estimates below $10^{-6}$ using the reweighted algorithms as a function of $c$  for the matrix in \eqref{eqn:rnd}.  Again, the gap in the plot is predicted by the arguments at the end of Section \ref{sec:gcovers}.}
	\label{fig:rnd}
\end{figure}

Although both the synchronous and asynchronous algorithms converge for the entire positive definite region in the above example, the synchronous and asynchronous algorithms can have very different convergence properties and damping may be required in order to force the synchronous algorithm to converge over arbitrary graphs, even for sufficiently large $c$.  Figure \ref{fig:rnd} illustrates these convergence issues for the matrix,
\begin{eqnarray}
\begin{pmatrix}
    45 &   21 &   23 &  -42\\
    21 &   83 &    8 &  -32\\
    23 &    8 &   14 &  -29\\
   -42 &  -32 &  -29 &  134\\
\end{pmatrix}.\label{eqn:rnd}
\end{eqnarray}
The above matrix was randomly generated.  Similar observations can be made for many other positive definite matrices as well.

\section{Conclusions and Future Research}
\label{sec:conc}
In this work, we explored the properties of reweighted message-passing algorithms for the quadratic minimization problem.  Our motivation was to address the convergence issues in the GaBP algorithm by leveraging the reweighting.  To this end, we employed graph covers to prove that standard approaches to convergence and correctness that exploit duality and coordinate ascent/descent such as MPLP \citep{MPLP}, tree-reweighted max-product \citep{waintrbp}, and \citet{sontag} are doomed to fail outside of walk-summable models.  While the GaBP variances may not converge outside of walk-summable matrices, we showed that there always exists a choice of reweighing parameters that guarantees monotone convergence of the variances.  Empirically, a similar strategy seems to guarantee convergence of the means as well.  As a result, our approach demonstrably outperforms the GaBP algorithm on this problem. We conclude this work with a discussion of a few open problems and directions for future research.

\subsection{Convergence}
The main open questions surrounding the performance of the reweighted algorithm relate to questions of convergence.  First, for all positive definite $\Gamma$, we conjecture that there exists a sufficiently large (or sufficiently negative) choice of the parameters such that the means always converge.

Second, in practice, one typically uses a damped version of the message updates in order to attempt to force convergence.  For the min-sum algorithm, the damped updates are given by
\begin{eqnarray*}
m^t_{i\rightarrow j}(x_j) & = & \kappa + \delta m^t_{i\rightarrow j}(x_j) + (1-\delta)\Big[\min_{x_i} \psi_{ij}(x_i,
x_j) + \phi_i(x_i)+\sum_{k \in \partial i \setminus j}m^{t-1}_{k\rightarrow i}(x_i)\Big].
\end{eqnarray*}

The damped min-sum algorithm with damping factor $\delta = 1/2$ empirically seems to converge if $\Gamma$ is positive definite and all of the computation trees remain positive definite \citep{malioutov}.  We make the same observation for the damped version of Algorithm \ref{alg:pairs}.  

In practice, the damped synchronous algorithm with $\delta=1/2$ and the asynchronous algorithm appear to converge for all sufficiently large choices of the parameter vector as long as $\Gamma$ is positive definite.  We conjecture that this is indeed the case:  for all positive definite $\Gamma$ there exists a $c$ such that if $c_{ij} = c $ for all $i\neq j$, then the asynchronous algorithm always converges.  In this line of exploration, the relationship between the synchronous and the asynchronous algorithms described in Section \ref{sec:sync} may be helpful.

Finally, \citet{convexciamac} were able to provide rates of convergence in the case that $\Gamma$ is walk-summable by using a careful analysis of the computation trees.  Perhaps similar ideas could be adapted for the computation trees produced by the reweigthed algorithm.  

\subsection{General Convex Minimization}

The previous graph cover observations can, in theory, be applied to minimize general convex functions, but in practice, computing and storing the message vector may be inefficient.  Despite this, many of the previous observations can be extended to any convex function $f:C\rightarrow\mathbb{R}$ such that $C\subseteq\mathbb{R}^n$ is a convex set.  

As was the case for quadratic minimization, convexity of the objective function $f^G$ does not necessarily guarantee convexity of the objective function $f^H$ for every finite cover $H$ of $G$.  Recall that the existence of graph covers that are not bounded from below can be problematic for the reweighted message-passing algorithm.  For quadratic functions, this cannot occur if the matrix is scaled diagonally dominant or, equivalently, if the objective function corresponding to every finite graph cover is positive definite.  This equivalence suggests a generalization of scaled diagonal dominance for arbitrary convex functions based on the convexity of their graph covers.  Such convex functions would have desirable properties with respect to iterative message-passing schemes.

\begin{lemma}
Let $f$ be a convex function that factorizes over a graph $G$. Suppose that for every finite cover $H$ of $G$, $f^H$ is convex.   If $x^G\in\arg\min_x f(x)$, then for every finite cover $H$ of $G$, $x^H$, the lift of $x^G$ to $H$, minimizes $f^H$.
\end{lemma}
\begin{proof}
This follows from the observation that all convex functions are subdifferentiable over their domains and that $x^H$ is a minimum of $f^H$ if and only if the zero vector is contained in the subgradient of $f^H$ at $x^H$.  
\end{proof}

Even if the objective function is not convex for some cover, we may still be able to use the same trick as in Theorem \ref{posdefcomp} in order to force the computation trees to be convex.  Let $C\subseteq \mathbb{R}^n$ be a convex set.  If $f:C\rightarrow \mathbb{R}$ is twice continuously differentiable, then $f$ is convex if and only if its Hessian, the matrix of second partial derivatives, is positive semidefinite on the interior of $C$.  For each fixed $x\in C$, Theorem \ref{posdefcomp} demonstrates that there exists a choice of the vector $c$ such that all of the computation trees are convex at $x$, but it does not guarantee the existence of a $c$ that is independent of $x$.  

For twice continuously differentiable functions, sufficient conditions for the convergence of the min-sum algorithm that are based on a generalization of scaled diagonal dominance are discussed in \citet{convexciamac}, and extending the above ideas is the subject of future research.

\appendix
\section{Proof of Theorem \ref{2cover}}
\label{ap:2cover}
Without loss of generality, we can assume that $\Gamma$ has a unit diagonal.  We break the proof into several pieces:
\begin{itemize}
\item $(1 \Rightarrow 2)$ Without loss of generality we can assume that $|I-\Gamma|$ is irreducible (if not we can make this argument on each of its connected components).  Let $1 > \lambda > 0$ be an eigenvalue of $|I - \Gamma|$ with eigenvector $x > 0$ whose existence is guaranteed by the Perron-Frobenius theorem. For any row $i$, we have:
\begin{eqnarray*}
x_i > \lambda x_i  =  \sum_{j\neq i} |\Gamma_{ij}| x_j.
\end{eqnarray*}
Since $\Gamma_{ii} = 1$ this is the definition of scaled diagonal dominance with $w = x$.

\item $(2 \Rightarrow 3)$ If $\Gamma$ is scaled diagonally dominant then so is every one of its covers.  Scaled diagonal dominance of a symmetric matrix with a positive diagonal implies that the matrix is symmetric positive definite.  Therefore, all covers must be symmetric positive definite.

\item $(3 \Rightarrow 4)$ Trivial.

\item $(4 \Rightarrow 1)$ Let $\widetilde{\Gamma}$ be any 2-cover of $\Gamma$.  Without loss of generality, we can assume that $\widetilde{\Gamma}$ has the form \eqref{eq:perm}.

First, observe that by the Perron-Frobenius theorem there exists an eigenvector $x > 0\in\mathbb{R}^n$ of $|I - \Gamma|$ with eigenvalue $\varrho(|I - \Gamma|)$.  Let $y\in\mathbb{R}^{2n}$ be constructed by duplicating the values of $x$ so that $y_{2i} = y_{2i+1} = x_i$ for each $i\in\{0...n\}$. By Lemma \ref{eigen}, $y$ is an eigenvector of $|I-\widetilde{\Gamma}|$ with eigenvalue equal to $\varrho(|I - \Gamma|)$. We claim that this implies $\varrho(|I - \widetilde{\Gamma}|)=\varrho(|I - \Gamma|)$. Assume without loss of generality that $|I-\widetilde{\Gamma}|$ is irreducible; if not, then we can apply the following argument to each connected component of  $|I-\widetilde{\Gamma}|$. By the Perron-Frobenius theorem again, $|I - \widetilde{\Gamma}|$ has a unique positive eigenvector (up to scalar multiple), with eigenvalue equal to the spectral radius. Thus, $\varrho(|I - \Gamma|)=\varrho(|I - \widetilde{\Gamma|})$ because $y>0$. 

We will now construct a specific cover $\widetilde{\Gamma}$ such that $\widetilde{\Gamma}$ is positive definite if and only if $\Gamma$ is walk-summable.  To do this, we'll choose the $P_{ij}$ as in \eqref{eq:perm} such that $P_{ij} = I$ if $\Gamma_{ij} < 0$ and $P_{ij} = \begin{pmatrix}0 & 1\\1 & 0\end{pmatrix}$ otherwise. Now define $z\in\mathbb{R}^{2n}$ by setting $z_i = (-1)^icy_i$, where the constant $c$ ensures that $\left\|z\right\| = 1$.  

Consider the following:
\begin{eqnarray*}
z^T\widetilde{\Gamma}z & = & \sum_{i=1}^n\sum_{j \neq i} \Gamma_{ij}[z_{2i} , z_{2i+1}]P_{ij}\begin{bmatrix}z_{2j}\\z_{2j+1}\end{bmatrix}+ \sum_i\Gamma_{ii} z_i^2\\
& = & 1 - 2 \sum_{i > j} |\Gamma_{ij}|c^2y_iy_j.
\end{eqnarray*}
Recall that $y$ is the eigenvector of $|I - \widetilde{\Gamma}|$ corresponding to the largest eigenvalue and $\left\|cy\right\| = 1$. By definition and the above,
\begin{eqnarray*}
\varrho(|I - \Gamma|) & = & \varrho(|I - \widetilde{\Gamma}|)\\
& = & \frac{cy^T|I - \widetilde{\Gamma}|cy}{c^2y^Ty}\\
& = & 2 \sum_{i > j} |\Gamma_{ij}|c^2y_iy_j.
\end{eqnarray*}
Combining all of the above we see that $z^T\widetilde{\Gamma}z = 1 - \varrho(|I - \Gamma|)$.  Now, $\widetilde{\Gamma}$ positive definite implies that $z^T\widetilde{\Gamma}z > 0$, so $1 - \varrho(|I - \Gamma|) > 0$.  In other words, $\Gamma$ is walk-summable.
\end{itemize}

\section{Proof of Theorem \ref{posdefcomp}}
\label{ap:posdefcomp}
Let $T_v(t)$ be the depth $t$ computation tree rooted at $v$, and let $\Gamma'$ be the matrix corresponding to $T_v(t)$ (i.e., the matrix generated by the potentials in the computation tree).  We will show that the eigenvalues of $\Gamma'$ are bounded from below by some $\epsilon > 0$.  For any $i\in T_v(t)$ at depth $d$ define:
\begin{eqnarray*}
w_i & = & \Big(\frac{s}{r}\Big)^d,
\end{eqnarray*}
where $r$ is as in the statement of the theorem and $s$ is a positive real to be determined below.  Let $W$ be a diagonal matrix whose entries are given by the vector $w$.  By the Ger\v{s}gorin disc theorem \citep{horn}, all of the eigenvalues of $W^{-1}\Gamma'W$ are contained in
\begin{eqnarray*}
\cup_{i\in T_v(t)} \Big\{ z\in \mathbb{R} : |z - \Gamma'_{ii}| \leq \frac{1}{w_i} \sum_{j\neq i} w_j |\Gamma'_{ij}|\Big\}.
\end{eqnarray*}
Because all of the eigenvalues are contained in these discs, we need to show that there is a choice of $s$ and $r$ such that for all $i\in T_v(t)$, $|\Gamma'_{ii}| - \frac{1}{w_i} \sum_{j\neq i} w_j |\Gamma'_{ij}| \geq \epsilon$.  

Recall from Section \ref{newcompsec} that $|\Gamma'_{ij}| = \eta\frac{|\Gamma_{ij}|}{r}$ for some constant $\eta$ that depends on $r$.  Further, all potentials below the potential on the edge $(i,j)$ are multiplied by $\eta\gamma$ for some constant $\gamma$.  We can divide out by this common constant to obtain equations that depend on $r$ and the elements of $\Gamma$.  Note that some self-potentials will be multiplied by $r-1$ while others will be multiplied by $r$.  With this rewriting, there are three possibilities:
\begin{enumerate}
\item $i$ is a leaf of $T_v(t)$.  In this case, we need $|\Gamma_{ii}|  > \frac{1}{w_i}\frac{|\Gamma_{ip(i)}|}{r} w_{p(i)}$.  Plugging in the definition of $w_i$, we have
\begin{eqnarray}
|\Gamma_{ii}|  > \frac{|\Gamma_{ip(i)}|}{s} \label{eq:leaf}.
\end{eqnarray}

\item $i$ is not a leaf of $T_v(t)$ or the root.  In this case, we need
\begin{eqnarray*}
|\Gamma_{ii}| & > & \frac{1}{w_i}\Big[\frac{|\Gamma_{ip(i)}|}{r} w_{p(i)} + \frac{s^2(r-1)}{r^3}|\Gamma_{ip(i)}| w_{p(i)}+\sum_{k\in\partial i - p(i)} |\Gamma_{ki}|w_k \Big].
\end{eqnarray*}
Again, plugging the definition of $w_i$ into the above yields
\begin{eqnarray*}
|\Gamma'_{ii}| & > & \frac{|\Gamma_{ip(i)}|}{s} + \frac{s}{r}\Big[\frac{r-1}{r}|\Gamma_{ip(i)}|+\sum_{k\in\partial i - p(i)} |\Gamma_{ki}| \Big].
\end{eqnarray*}
\item $i$ is the root of $T_v(t)$.  Similar to the previous case, we need $|\Gamma_{ii}| w_i > \sum_{k\in\partial i} |\Gamma_{ki}|w_k$.  Again, plugging the definition of $w_i$ into the above yields
\begin{eqnarray*}
|\Gamma_{ii}| > \frac{s}{r}\sum_{k\in\partial i} |\Gamma_{ki}|.
\end{eqnarray*}
\end{enumerate}
  
None of these bounds are time dependent.  As such, if we choose $s$ and $r$ to satisfy the above constraints, then there must exist some $\epsilon > 0$ such that smallest eigenvalue of any computation tree is at least $\epsilon$.  Fix $s$ to satisfy \eqref{eq:leaf} for all leaves of $T_v(t)$.  This implies that $(|\Gamma_{ii}| - \frac{|\Gamma_{ip(i)}|}{s}) > 0$ for any $i \in T_v(t)$.  Finally, we can choose a sufficiently large $r$ that satisfies the remaining two cases for all $i\in T_v(t)$.\\

\bibliography{biblio}

\end{document}